\documentclass[12pt]{article}

\usepackage{amsmath,amssymb,amsthm,cite,enumitem,mathtools,xcolor}
\usepackage[margin=3cm]{geometry}

\usepackage[T1]{fontenc}
\usepackage[osf]{newtxtext} 
\usepackage[nosymbolsc,cmintegrals]{newtxmath}
\usepackage[cal=euler,frak=euler]{mathalfa} 
\usepackage{bm} 

\title{Revisiting the Okubo--Marshak argument}

\author{Christian Ga\ss,$^1$
José M. Gracia-Bondía$^{2,3}$%
\footnote{Email: jmgb@unizar.es}
\ and Jens Mund$^4$%
\\[6pt]
{\footnotesize $^1$ Institute for Theoretical Physics, 
Georg-August-University Göttingen, 37077 Göttingen, Germany}
\\[3pt]
{\footnotesize $^2$ CAPA and Departamento de Física Teórica,
Universidad de Zaragoza, Zaragoza 50009, Spain}
\\[3pt]
{\footnotesize $^3$ Laboratorio de Física Teórica y Computacional,
Universidad de Costa Rica, San Pedro 11501, Costa Rica}
\\[3pt]
{\footnotesize $^4$ Departamento de Física, 
Universidade Federal de Juiz de Fora, Juiz de Fora 36036-900, MG,
Brasil}
}

\date{\today}

\DeclareMathOperator{\T}{T}         
\DeclareMathOperator{\tsum}{{\textstyle\sum}} 

\newcommand{\dl}{\delta}            
\newcommand{\eps}{\varepsilon}      
\newcommand{\ka}{\kappa}            
\newcommand{\La}{\Lambda}           
\newcommand{\la}{\lambda}           
\newcommand{\sg}{\sigma}            
\newcommand{\vf}{\varphi}           

\newcommand{\bD}{\mathbb{D}}        
\newcommand{\bM}{\mathbb{M}}        
\newcommand{\bR}{\mathbb{R}}        
\newcommand{\bS}{\mathbb{S}}        

\newcommand{\sD}{\mathcal{D}}       
\newcommand{\sO}{\mathcal{O}}       
\newcommand{\sR}{\mathcal{R}}       
\newcommand{\sS}{\mathcal{S}}       

\newcommand{\dv}{{\mathrm{div}}}    
\newcommand{\head}{\mathrm{H}}      
\newcommand{\rK}{\mathrm{K}}        
\newcommand{\rL}{\mathrm{L}}        
\newcommand{\sym}{{\mathrm{sym}}}   
\newcommand{\tail}{\mathrm{T}}      
\newcommand{\tree}{{\mathrm{tree}}} 

\bmdefine{\ee}{e}             
\bmdefine{\pp}{p}           

\newcommand{\del}{\partial}         
\newcommand{\downto}{\downarrow}    
\newcommand{\otto}{\leftrightarrow} 
\newcommand{\ovl}{\overline}        
\newcommand{\upto}{\uparrow}        
\newcommand{\7}{\dagger}            
\newcommand{\8}{\bullet}            

\newcommand{\half}{{\mathchoice{\thalf}{\thalf}{\shalf}{\shalf}}}
\newcommand{\shalf}{{\scriptstyle\frac{1}{2}}} 
\newcommand{\thalf}{\tfrac{1}{2}}   

\newcommand{\ket}[1]{\,\lvert#1\rangle} 
\newcommand{\set}[1]{\{\,#1\,\}}  
\newcommand{\vev}[1]{\langle\!\langle#1\rangle\!\rangle} 
\newcommand{\word}[1]{\quad\text{#1}\quad} 

\def\wick:#1:{\,\mathopen:#1\mathclose:\,} 

\newcommand{\braket}[2]{\langle#1\mathbin|#2\rangle} 

\theoremstyle{plain}
\newtheorem{prop}{Proposition}      

\makeatletter
\renewcommand{\section}{\@startsection{section}{1}{\z@}%
                       {-3.5ex \@plus -1ex \@minus -.2ex}%
                       {2.3ex \@plus.2ex}%
                       {\normalfont\large\bfseries}}
\renewcommand{\subsection}{\@startsection{subsection}{2}{\z@}%
                       {-3.25ex \@plus -1ex \@minus -.2ex}%
                       {1.5ex \@plus .2ex}%
                       {\normalfont\normalsize\bfseries}}
\makeatother

\begin{document}

\maketitle

\begin{flushright}
\textit{To the memory of Bob Marshak and Daniel Testard}
\end{flushright}

\medskip

\begin{abstract}
Modular localization and the theory of string-localized fields have
revolutionized several key aspects of quantum field theory. They
reinforce the contention that \textit{local} symmetry emerges directly
from quantum theory, but global gauge invariance remains in general an
unwarranted assumption, to be examined case by case. Armed with those
modern tools, we reconsider here the classical Okubo--Marshak argument
on the non-existence of a ``strong CP problem'' in quantum
chromodynamics.
\end{abstract}

\medskip

\textit{Keywords}:
String independence, local symmetry, strong CP problem


\section{Introduction: string-localized fields}
\label{sec:quod-raro-fit}

In this paper the case by Okubo and Marshak~\cite{OM92} against
existence of the ``strong CP problem'' in quantum chromodynamics,
which was based on the covariant approach to Yang--Mills theory by
Kugo and Ojima~\cite{VerdaderoKO}, is reassessed from a different
theoretical standpoint. For the purpose we bring to bear the theory of
string-localized quantum fields (SLF).

Those are not a recent invention. The paper by Dirac
\cite{CanadianDirac} should be regarded as a precedent. Charged SLF,
recognizably similar to the modern version, were a brainchild of
Mandelstam~\cite{Mandelstam62}. In a different vein, they were treated
by Buchholz and Fredenhagen~\cite{BF82}. In the eighties they were
further developed by Steinmann
\cite{Steinmann82,Steinmann84,Steinmann85}. They resurfaced over 15
years ago in important papers by one of us, Schroer and
Yngvason~\cite{MundSY04,MundSY06}. Interest in SLF had meanwhile been
sustained and renewed by their obvious connection to modular
structures~\cite{Borchers2000}, with roots in deep mathematical
results of Hilbert algebra theory~\cite{Alain} and a geometrical
interpretation through the Bisognano--Wichmann relations
\cite{BW76,BWJens}. Among the harbingers of their revival we count
papers \cite{DuetschS00,FassarellaS02,BGL02}, concerned with the
proper concept of locality in modern quantum field theory.

The gist of~\cite{MundSY04} was to show that within the SLF
dispensation the Wigner particles of mass zero and unbounded helicity
possess associated quantum fields; those had long before been excluded
from the standard framework, by Yngvason himself~\cite{Yngvason70}.
The detailed treatment in \cite{MundSY06}, apart from reviewing this
fact, directly and comprehensively relates SLF to point-localized
fields of the ordinary sort.

At the price of an extra variable, SL fields offer important
advantages. Two of these are: (a)~String-local fields slip past the
theorem \cite[Sect.~5.9]{Weinberg95I} that it is impossible to
construct on Hilbert space tensor fields of rank~$r$ for massless
particles with helicity~$r$ -- like photons and gluons. (b)~Their
short-distance behaviour, both for massive and massless particles, is
the \textit{same for all bosons} as for scalars, and for all fermions
as for spin-$\half$ particles.%
\footnote{\textit{Free} scalar and spin-$\half$ particles are
non-stringy.}

The primary upshot is that large no-go territories for QFT are now
trespassed. The separation of helicities in the massless limit of
higher spin fields is clarified~\cite{MRS17b}. The
van~Dam-Veltman--Zakharov discontinuity~\cite{vDV70,Zh70} at the
$m \downto 0$ limit of massive gravitons is resolved~\cite{MRS17a}.
Unimpeachable stress–energy-momentum (SEM) tensors for massless fields
of \textit{any} helicity are constructed~\cite{MRS17b,MRS17a} --
allowing for gravitational interaction, and flouting the
Weinberg--Witten theorems in particular. A good candidate SEM for
unbounded helicity particles now exists in SL field theory,
too~\cite{RehrenPLLimit}. The Velo--Zwanziger instability~\cite{VZ69}
is exorcised~\cite{Schroernuclear}. The prohibition for covariant
spinorial field types $(A,B)$ to represent massless Wigner particles
of helicity~$r$, unless $r = B - A$, argued in
\cite[Sect.~5.9]{Weinberg95I}, is made void. SL~field theory is also
helping to deal with profound, age-old problems of QED~\cite{MRS20}.

In summary, SL fields sit comfortably midway between ``ordinary'' and
algebraic QFT. More to our point here, renormalization of QFT models
is to take place without calling upon ghost fields, BRS invariance and
the like -- since for SL fields one need not surrender positivity
neither of the energy, nor of the state spaces for the physical
particles.

Regarding the Standard Model, rigorous perturbative field theory
\cite{EpsteinGlaser73} together with the \textit{principle of string
independence} of the $\bS$-matrix allowed us to show some time ago
that chirality of the electroweak interactions does not have to be put
by hand, as in the GWS model. Rather, it follows ineluctably from the
massive character of their interaction carriers -- a conclusion
irrespective of ``mechanisms'' conjuring the mass~\cite{Rosalind}. The
proof built on outstanding work by Aste, Dütsch and Scharf at the turn
of the century, on a \textit{quantum} formulation for gauge
invariance~\cite{AS99,DuetschS99}.%
\footnote{See \cite{Grigore2000} as well.}
Following Marshak, we refer to electroweak theory as quantum
flavourdynamics (QFD) henceforth.

\medskip

SLF have their own drawbacks, to be sure. Practical calculations of
loop graphs with SL fields in internal lines are rather challenging,
and a proof of normalizability at all orders is still pending. The
applications of SLF so far concern mostly questions of principle,
badly dealt with within conventional QFT. Of which there are plenty.
This paper addresses one of them, that can be handled by means of tree
graphs.

\subsection{Massless bosonic SLF: general theory}
\label{ssc:ad-hoc}

The term ``string''%
\footnote{Not to be confused with the strings of string theory.}
in the present context precisely denotes a ray starting at a point~$x$
in Minkowski space $\bM_4$ that extends to infinity in a spacelike or
lightlike direction. This is the natural limit of the spacelike cones
in the intrinsic localization procedure of~\cite{BGL02}. The set of
such strings can be parametrized by the one-sheet hyperboloid
$H_3 \subsetneq \bM_4$ (``de Sitter space'') of neck radius equal
to~one.%
\footnote{Use of the limiting case of lightlike strings, parametrized
by the celestial spheres $S^2 \cup S^2 \subsetneq \bM_4$, simplifies
some formulae. They are a little troublesome from the
functional-analytic viewpoint, however.}

By way of example, let us look first at the ``Abelian'' case of
helicity $r = 1$. Let $d\mu(p) = (2\pi)^{-3/2}\,d^3\pp/2|\pp|$. The
quantum electromagnetic field strength, built on Wigner's
helicity~$\pm 1$ unirreps of the Poincaré group, is written:
\begin{equation}
F_{\mu\nu}(x) = \sum_{\sg=\pm} \int d\mu(p)\, \bigl[
e^{i(px)} u^\sg_{\mu\nu}(p) a^\7(p,\sg)
+ e^{-i(px)} \bar u^\sg_{\mu\nu}(p) a(p,\sg) \bigr],
\end{equation}
where $(px)\equiv g_{\mu\nu} p^\mu x^\nu$ with mostly negative metric
$(g_{\mu\nu})$; the intertwiners are of the form 
$u^\sg_{\mu\nu}(p) = i e^\sg_\nu(p) p_\mu - i e^\sg_\mu(p) p_\nu$, for
$e^\pm$ a polarization zweibein, with $\bar u$ denoting the complex
conjugate of~$u$.

The corresponding string-localized vector field for photons is given
by:
\begin{subequations}
\begin{gather}
A_\mu(x,e) = \sum_{\sg=\pm} \int d\mu(p)\, \bigl[
e^{i(px)} u^\sg_\mu(p,e) a^\7(p,\sg) 
+ e^{-i(px)} \bar u^\sg_\mu(p,e) a(p,\sg) \bigr],
\\
\text{its intertwiners } u^\sg_\mu(p,e) \text{ being of the form}\quad
\frac{(pe) e^\pm_\mu(p) - (e^\pm(p)e) p_\mu}{(pe)}.
\end{gather}
\end{subequations}

The denominator $(pe)$ in $u^\sg_\mu(p,e)$ is shorthand for the
distribution $\bigl((pe) + i0 \bigr)^{-1}$, to be smeared in the
string variable~$e$. The following key relations, respectively
integral and differential, are effortlessly derived:
\begin{align}
A_\mu(x,e) = \int_0^\infty dt\, F(x + te)_{\mu\nu}\, e^\nu;  \quad
\del_\mu A_\nu(x,e) - \del_\nu A_\mu(x,e) = F_{\mu\nu}(x);
\label{eq:kaizen} 
\end{align}
so that $A_\mu(x,e)$ is a \textit{bona fide} potential for
$F_{\mu\nu}(x)$. For the first identity:
\begin{equation}
i\bigl( p_\mu e^\pm_\la(p) - p_\la e^\pm_\mu(p) \bigr) e^\la
\lim_{\eps\downto 0} \int_0^\infty dt\, e^{it((pe)+i\eps)}
= e^\pm_\mu(p) - p^\mu\,\frac{(e^\pm(p)\,e)}{(pe)}.
\end{equation}
It ought to be clear that the vector potential $A_\mu(x,e)$ fulfils
the equations $e^\mu A_\mu(x,e) = \del^\mu A_\mu(x,e) = 0$. These
``transversality'' properties are necessary for the free field
$A_\mu(x,e)$ acting on the physical Hilbert space
\cite[Sect.~5]{MundSY06}. Their role is to reduce the number of
degrees of freedom, as required for on-shell photons -- in much the
same way as the six components of the electromagnetic field reduced by
the Maxwell equations propagate two degrees of freedom.%
\footnote{Only the second could perhaps be taken as a Lorenz ``gauge
condition''.}

The reader should appreciate fully the deep differences between
$A_\mu(x,e)$ and the $A_\mu(x)$ potentials of standard QFT. In 
particular, there are no ``pure gauge'' configurations in QED or QCD,%
\footnote{Understood in this paper in the narrow sense of
gluodynamics, the theory of pure Yang--Mills fields.}
when working in SL field theory. The field $A_\mu(x,e)$ lives on the
same Fock space as~$F_{\mu\nu}$. Thus the second equation in
\eqref{eq:kaizen} is an operator relation; which is not the case for
the similar one in the usual framework. All this makes easier the
physical interpretation of the present one.

In contrast with the standard formalism, $A_\mu(x,e)$ is perfectly
covariant: for $\La$ a Lorentz transformation, $c$ a translation, and
$U$ the second quantization of the mentioned unirrep pair of the
Poincaré group:
\begin{equation}
U(c,\La) A_\mu(x,e) U^\7(c,\La) 
= (\La^{-1})_\mu^{\,\la} A_\la(\La x + c, \La e).
\end{equation}

It is important to realize that the string-differential of the photon 
field is a gradient:
\begin{align}
d_e A_\mu(x,e) 
&:= \sum_\rho de^\rho \frac{\del A_\mu(x,e)}{\del e^\rho}
= -\sum_{\sg=\pm} \int d\mu(p)\, \biggl[ 
e^{i(px)} \biggl( \frac{p_\mu e_\rho^\sg}{(pe)} 
- \frac{p_\rho p_\mu (e^\sg e)}{(pe)^2} \biggr) a^\7(p,\sg)
\notag \\
&\qquad
+ e^{-i(px)} \biggl( \frac{p_\mu e_\rho^\sg}{(pe)} 
- \frac{p_\rho p_\mu(e^\sg e)}{(pe)^2} \biggr)^- a(p,\sg) \biggr]
\,de^\rho
\notag \\
&= i\del_\mu \sum_{\sg=\pm} \int d\mu(p)\, \biggl[
e^{i(px)} \biggl( \frac{e_\rho^\sg}{(pe)} 
- \frac{p_\rho (e^\sg e)}{(pe)^2} \biggr) a^\7(p,\sg)
\notag \\
&\qquad
- e^{-i(px)} \biggl( \frac{e_\rho^\sg}{(pe)} 
- \frac{p_\rho(e^\sg e)}{(pe)^2} \biggr)^- a(p,\sg) \biggr]
\,de^\rho =: \del_\mu u(x,e).
\end{align}
Naturally, $u$ satisfies the wave equation and 
$(d_e)^2 A_\mu = \del_\mu d_e u(x,e) = 0$.

Last, but not least, locality:
\begin{equation}
[A_\mu(x,e), A_\nu(x',e')] = 0
\end{equation}
holds whenever the strings $x + \bR^+ \tilde e$ and 
$x' + \bR^+ e'$ are spacelike separated for all $\tilde e$ in some
open neighbourhood of~$e$: ``causally disjoint''. A proof is found in
\cite[App.~C]{Rosalind}.

\section{Dealing with string independence}
\label{sec:non-observant}

Perturbation theory for the SLF is to be attacked here in the spirit
of ``renormalization without regularization'' of rigorous $\bS$-matrix
theory~\cite{EpsteinGlaser73}. The method engages the construction of
a Bogoliubov-type functional scattering operator $\bS[g;h]$ dependent
on a multiplet~$g$ of external fields and a test function
$h \in \sD(H_3)$ with integral~$1$ that averages over the string
directions. $\bS[g;h]$ is submitted to the customary conditions of
causality, unitarity and covariance. One looks for it as a formal
power series in~$g$,
\begin{equation}
\bS[g;h] := 1 + \sum_{n=1}^\infty \frac{i^n}{n!} \prod_{k=1}^n 
\prod_{l=1}^m \int d^4x_k \int d\sg(e_{k,l})\, g(x_k) h(e_{k,l})
S_n(x_1,\ee_1;\dots;x_n,\ee_n),
\end{equation}
where $\sg$ is the measure on~$H_3$. Only the first-order vertex
coupling $S_1 = S_1(x,\ee)$, a Wick polynomial in the free fields, is
postulated -- already under severe restrictions. It depends on an
array $\ee = (e_1,\dots,e_m)$ of string coordinates, with $m$ the
maximum number of SLF appearing in a sub-monomial of~$S_1$. For
$n \geq 2$, the $S_n$ are time-ordered products that need to be
constructed. Two sets of strings cannot be ordered, after chopping
them into segments if necessary~\cite{Atropos}, if and only if they
touch each other -- see the Appendix for the details. The resulting
exceptional set
\begin{equation}
\bD_2 := \set{(x,\ee;x',\ee') : 
(x + \bR^+ e_k) \cap (x' + \bR^+ e'_l) \neq \emptyset 
\text{ for some } k,l}
\label{eq:exceptional-set} 
\end{equation}
is a set of measure zero that includes the diagonal $x = x'$. A
similar statement holds for $n > 2$. Only that part of $S_n$
contributes to $\bS[g;h]$ which is symmetric under permuting the
string coordinates, which are smeared with the same test 
function~$h$.%
\footnote{This symmetry will be heavily exploited in the following
derivations.}
The extension of the $S_n$-products across $\bD_n$ is the
renormalization problem in a nutshell~\cite{Gass21}.

The natural and essential hypothesis of interacting SLF theory is
simple enough: physical observables and quantities closely related to
them, particularly the $\bS$-matrix, cannot depend on the string
coordinates. This is the intrinsically quantum \textbf{string
independence} principle: colloquially, the strings ``ought not to be
seen''. In this paper it will replace the ``gauge principle'' with
advantage.

For the physics of the model described by $S_1$ to be
string-independent, one must require that a vector field
$Q^\mu(x,\ee)$ exist such that, after appropriate symmetrization in
the string variables
\begin{equation}
d_{e_1} S_1^\sym = (\del Q) \equiv \del_\mu Q^\mu,
\end{equation}
so that on applying integration by parts in the ``adiabatic limit'',
as $g$ goes to a set of constants, the contribution from the
divergence vanishes. Then the covariant $\bS[g;h]$ approaches the
invariant physical scattering matrix~$\bS$, therefore
$U(a,\La) \bS\, U^\7(a,\La) = \bS$, all dependence on the strings
disappearing.

On the face of it, existence of the adiabatic limit is the property
that the $S_k$ be integrable distributions. Due to severe infrared
problems, the latter does not hold in QCD, which involves us here.
However, a recent breakthrough~\cite{Pawel} rigorously establishes
existence of a suitable weaker adiabatic limit (WAL) in~QCD, and so
the above reasoning can proceed.%
\footnote{The cited work is concerned with point-local fields. It can
be generalized to the string-local setting in low orders. A proof of
WAL at all orders in the SLF context is still awaiting.}

\subsection{The Aste--Scharf argument recast in SLF theory}
\label{ssc:ad-inquirendum}

\begin{prop} 
\label{pr:skewsymmetry}
Suppose that we are given $n$ massless fields $A_{\mu a}$,
$(a = 1,\dots,n)$. For their mutual cubic coupling modulo divergences,
string independence at first order enables the Wick product
combination:
\begin{equation}
S_1(x,e_1,e_2) = \frac{g}{2}\,
f_{abc}\, A_{\mu a}(x,e_1) A_{\nu b}(x,e_2) F^{\mu\nu}_c(x),
\label{eq:vetustas} 
\end{equation}
where the $f_{abc}$ are \textbf{completely skewsymmetric}
coefficients. (Subindices that appear twice are summed over.)
\end{prop}

Before proceeding, we note that this vertex promulgates a
\textit{renormalizable} theory by power counting. Note also that $S_1$
is intrinsically symmetric in the string coordinates, and that
\begin{equation}
d_{e_1} S_1 = \frac{g}{2}\, \del_\mu \bigl[ f_{abc}\, 
u_a(x,e_1) A_{\nu b}(x,e_2) F^{\mu\nu}_c(x) \bigr]
=: \del_\mu Q^\mu(x,e_1,e_2).
\end{equation}

\begin{proof}[Proof of Prop.~\ref{pr:skewsymmetry}]
We abbreviate $A^i \equiv A(x,e_i)$ for the SLF and make the Ansatz:  
\begin{equation}
S'_1(x,e_1,e_2,e_3)
= g f^1_{abc} A^1_{\mu a} A^2_{\nu b}\,\del^\mu A^{3\nu}_c
\label{eq:non-bis} 
\end{equation}
for the cubic coupling of the fields, the coefficients $f^1_{abc}$
being \textit{a~priori} unknown. We shall show that string
independence forces this to be of the form $S_1$ as in
\eqref{eq:vetustas}, where those numerical coefficients $f_{abc}$ are
completely skewsymmetric. We first split 
$f^1_{abc} =: d_{abc} + f^2_{abc}$ into a symmetric
($d_{abc} = d_{acb}$) and a skewsymmetric part 
($f^2_{abc} = - f^2_{acb}$) under exchange of the second and third
indices. After symmetrizing~\eqref{eq:non-bis} in $e_2 \otto e_3$, the
contribution of~$d_{abc}$ yields a divergence:
\begin{equation}
d_{abc} A^1_{\mu a} (A^2_{\nu b} \del^\mu A^{3\nu}_c 
+ A^3_{\nu b} \del^\mu A^{2\nu}_c) 
= \del^\mu(d_{abc}\ A^1_{\mu a} A^2_{\nu b} A^{3\nu}_c).
\end{equation}

We can therefore replace $S'_1(x,e_1,e_2,e_3)$ by
\begin{equation}
S''_1(x,e_1,e_2,e_3) 
:= g f^2_{abc} A^1_{\mu a} A^2_{\nu b} \del^\mu A^{3\nu}_c.
\label{eq:abusus} 
\end{equation}

Its symmetrized version satisfies 
\begin{equation}
d_{e_1} \tsum_{\pi\in S_3} S''_1(x,e_{\pi(1)},e_{\pi(2)},e_{\pi(3)})
= 2 g f^2_{abc} \bigl( \del_\nu A^2_{\mu a} \,\del^\mu A^{3\nu}_b
+ \del_\nu A^2_{\mu b} \,\del^\mu A^{3\nu}_a \bigr) u^1_c + \dv.
\label{eq:auribus-teneo-lupum} 
\end{equation}

We next split the coefficients $f^2_{abc} = f^+_{abc} + f^-_{abc}$
into symmetric and skewsymmetric parts under exchange of the first two
indices, $f^\pm_{abc} = \pm f^\pm_{bac}$. The skewsymmetric part
$f^-_{abc}$ does enter into~\eqref{eq:auribus-teneo-lupum}, whereas
the symmetric part $f^+_{abc}$ contributes
\begin{equation}
d_{e_1} \sum_{\pi\in S_3} S''_1(x,e_{\pi(1)},e_{\pi(2)},e_{\pi(3)})
= 4g f^+_{abc} \del_\nu A^2_{\mu a} \,\del^\mu A^{3\nu}_b u^1_c + \dv.
\end{equation}

By our basic postulate, this must be a divergence. Since the operators
$\del_\nu A^2_{\mu a}\,\del^\mu A^{3\nu}_b u^1_c$ are linearly
independent, that can happen iff the symmetric part $f^+_{abc}$ is
identically zero. This means complete skewsymmetry of the 
$f^2_{abc} \equiv f^-_{abc} =: f_{abc}$. That is to say, the string
independence principle constrains $S''_1$ in~\eqref{eq:abusus} to the
form
\begin{equation}
S'''_1 = g f_{abc} A^1_{\mu a} A^2_{\nu b} \del^\mu A^{3\nu}_c 
= \frac{g}{2}\, f_{abc}\, A^1_{\mu a} A^2_{\nu b} F_c^{\mu\nu} =: S_1,
\end{equation}
so that the dependence on~$e_3$ is trivial and
formula~\eqref{eq:vetustas} with the stated skewsymmetry condition is
established.
\end{proof}

It is worth pointing out that the above reasoning for the form
of~$S_1$ becomes simpler in our SLF context than in the quantum gauge
invariance approach~\cite[Sect.~3.1]{AS99}, the inference there
being in terms of the customary fields and their ungainly retinue of
unphysical ones.

\subsection{Dealing with string independence at second order:
preliminaries}
\label{ssc:mutare-consilium}

Perturbative string independence should hold at every order in the
couplings, surviving renormalization. Now, the $f_{abc}$ do not yet a
Lie algebra make; for that one needs to prove a \textit{Jacobi
identity}. This is going to be obtained from string independence in
constructing the functional $\bS$-matrix at second order in the
couplings.

Outside the exceptional set $\bD_2$ from~\eqref{eq:exceptional-set},
time-ordered products of string-local fields reduce to ordinary
products where the order of terms is determined by the geometric
time-ordering of the string segments. There, string variation and
derivatives commute with time ordering and we have
\begin{align}
d_{e_1} S_2(x,e_1,e_2;x',e_1',e_2') 
&= d_{e_1} \T\bigl( S_1(x,e_1,e_2) S_1(x',e'_1,e'_2) \bigr) 
\notag \\ 
&= \del_\mu \T\bigl( Q^\mu(x,e_1,e_2) S_1(x',e'_1,e'_2) \bigr).
\label{eq:multam-pecuniam-deportat} 
\end{align}

Above, $\T$ denotes a generic time-ordering recipe, that is, an
extension of the time ordering across~$\bD_2$. It does \textit{not}
automatically follow that~\eqref{eq:multam-pecuniam-deportat} holds
over the \textit{whole} $(x,e_1,e_2;x',e_1',e_2')$ set. The challenge
is to manufacture a time-ordered product $S_2$ so that this property
holds everywhere after appropriate symmetrization in the string
variables. The construction of~$S_2$ by solving the
\textit{obstructions} to such an identity will \textit{impose} the
couplings of ``non-Abelian gauge theory''. The vector quantity $Q^\mu$
will play a central role in our development.

Candidate extensions across $\bD_2$ are restricted by the requirement
that the Wick expansion hold everywhere. We are concerned only with
the tree graph for gluon-gluon scattering. Its corresponding
amplitude is of the general form
\begin{equation}  
\T(UV')_\tree = \sum_{\vf,\chi'} \frac{\del U}{\del\vf}
\vev{\T \vf\,\chi'}\, \frac{\del V'}{\del\chi'},
\label{eq:Wick-expand} 
\end{equation}
where $\vev{-}$ denotes a vacuum expectation value, the sum in the
brackets goes over all the free fields entering the monomials $U(x)$,
$V(x')$, and we employ formal derivation within the Wick polynomials.

For time-ordered products of the fields entering $S_2$ in our model,
we naturally consider in the first place:
\begin{equation}
\vev{\T_0 \vf(x,e)\,\chi(x',e')} := \frac{i}{(2\pi)^4} \int d^4p\,
\frac{e^{-i(p(x - x'))}}{p^2 + i0}\, M^{\vf\chi}(p,e,e'),
\label{eq:pandemonium} 
\end{equation}
where the $M^{\vf\chi}(p,e,e')$ are given by
\begin{equation}
M^{\vf\chi}_{*\8} 
:= \sum_\sg \ovl{u^{\sg;\vf}_*(p,e)}\, u^{\sg;\chi}_\8(p,e'),
\end{equation}
for the appropriate spacetime indices $*,\8$, in terms of the
respective intertwiners. We need:
\begin{subequations}
\label{eq:far-and-few} 
\begin{align}
M^{FA}_{\mu\nu,\la}(p,e')
&= i\biggl( p_\mu g_{\nu\la} - p_\nu g_{\mu\la} 
+ p_\la \frac{p_\nu e'_\mu - p_\mu e'_\nu}{(pe') + i0} \biggr),
\label{eq:sed-tamen} 
\\
M^{FF}_{\mu\nu,\ka\la}(p)
&= p_\nu p_\ka \,g_{\mu\la} + p_\mu p_\la \,g_{\nu\ka} 
- p_\nu p_\la \,g_{\mu\ka} - p_\mu p_\ka \,g_{\nu\la},
\label{eq:utique-fatetur} 
\end{align}
\end{subequations}
to be found for instance in our~\cite{Rosalind}. Note that in all
generality
\begin{equation}
d_e \vev{\T_0 \vf(x,e)\, \chi(x',e')} 
= \vev{\T_0 d_e\vf(x,e)\, \chi(x',e')},
\end{equation}
and similarly for~$d_{e'}$. 

The problem of resolving the obstructions to
Eq.~\eqref{eq:multam-pecuniam-deportat} will be reduced to an
extension problem for numerical distributions by carefully
constructing the contractions~\eqref{eq:pandemonium}. At present we
have, with completely skewsymmetric coefficients~$f_{abc}$:
\begin{equation}
S_1 = \frac{g}{2} f_{abc}\, A^1_{\mu a} A^2_{\nu b} F^{\mu\nu}_c
\word{and}  Q^\mu(x,e_1,e_2) 
= \frac{g}{2} f_{abc}\, u^1_a A^2_{\nu b} F^{\mu\nu}_c.
\end{equation}
Clearly $S_1$ is symmetric under exchange of the string variables. We
want $S_2(x,e_1,e_2;x',e'_1,e'_2)$ to be symmetric under exchange of
$(x,e_1,e_2) \otto (x',e'_1,e'_2)$. Therefore, resolving all
obstructions to Eq.~\eqref{eq:multam-pecuniam-deportat} is equivalent
to inspecting the obstruction
\begin{align}
\sO^1 
&:= d_{e_1} \bigl( \T_0(S_1(x,e_1,e_2) S_1(x',e'_1,e'_2)) \bigr)_\tree
- \del_\mu \T_0\bigl( Q^\mu(x,e_1,e_2) S_1(x',e'_1,e'_2) \bigr)_\tree
\notag \\
&= \T_0\bigl( \del_\mu Q^\mu(x,e_1,e_2) S_1(x',e'_1,e'_2) \bigr)_\tree
- \del_\mu \T_0\bigl( Q^\mu(x,e_1,e_2) S_1(x',e'_1,e'_2) \bigr)_\tree,
\label{eq:existimatione-mentis} 
\end{align}
to investigate how it can be made to vanish after appropriate
symmetrization. To the purpose, with the delta-function $\dl_e$ along
the string~$e$ defined as
\begin{equation}
\dl_e(x - x') := \int_0^\infty dt\, \dl(x + te - x'),
\end{equation}
one obtains, for similarly coloured gluons:
\begin{subequations}
\label{eq:plus-in-re} 
\begin{align}
\del_\mu \vev{\T_0(F^{\mu\nu}(x) A^\la(x',e))}
&= i \bigl[ g^{\nu\la} \dl(x - x') 
- e^\nu \del^\la \dl_{-e}(x - x') \bigr];
\label{eq:plus-in-re-est} 
\\
\del_\mu \vev{ \T_0(F^{\mu\nu}(x) F^{\ka\la}(x'))} 
&= i(g^{\nu\ka} \del^\la - g^{\nu\la} \del^\ka) \,\dl(x - x').
\label{eq:plus-in-re-quam} 
\end{align}
\end{subequations}
These formulae are easily derived by use of Eqs.\
\eqref{eq:pandemonium}, \eqref{eq:sed-tamen}
and~\eqref{eq:utique-fatetur}, respectively.

\subsection{The Jacobi identity emerges}
\label{ssc:causa-perit}

We compute, with an eye on~\eqref{eq:Wick-expand} and another 
on~\eqref{eq:existimatione-mentis}:
\begin{align}
\sO^1 &= -\frac{g}{2} f_{abc} \sum_{\chi'} u^1_a A^2_{\nu b}  
\del_\mu \vev{\T_0 F^{\mu\nu}_c \chi'}
\frac{\del S_1(x',e'_1,e'_2)}{\del\chi'} 
\notag \\
&= - \frac{g^2}{4} f_{abc} f_{dfc} u^1_a A^2_{\nu b} \bigl[ 
\del_\mu \vev{\T_0 F^{\mu\nu} F'^{\ka\la}} 
A'^{1'}_{\ka d} A'^{2'}_{\la f}
\notag \\
&\qquad + \del_\mu \vev{\T_0 F^{\mu\nu} A'^{1'}_\ka} A'^{2'}_{\la d}
F'^{\ka\la}_f 
- \del_\mu \vev{\T_0 F^{\mu\nu} A'^{2'}_\la} A'^{1'}_{\ka d}
F'^{\ka\la}_f \bigr],
\end{align}
where $A'^{1'}_{\ka d} \equiv A_{\ka d}(x',e'_1)$, 
$A'^{2'}_{\la f} \equiv A_{\la f}(x',e'_2)$, and similarly.
Employing~\eqref{eq:plus-in-re-est} and~\eqref{eq:plus-in-re-quam}, 
this obstruction equals 
\begin{align}
- i\frac{g^2}{4} f_{abc} f_{dfc} u^1_a A^2_{\nu b} 
&\bigl[ A^{1'}_{\ka d} A^{2'}_{\la f}
(g^{\nu\ka} \del^\la - g^{\nu\la} \del^\ka) \,\dl(x - x') 
\notag \\
&\enspace + A'^{2'}_{\la d} F'^{\ka\la}_f \bigl( g^\nu_\ka\,\dl(x - x')
- e'^\nu_1 \del_\ka \dl_{-e'_1}(x - x') \bigr)
\notag \\
&\enspace - A'^{1'}_{\ka d} F'^{\ka\la}_f \bigl( g^\nu_\la\,\dl(x - x')
- e'^\nu_2 \del_\la \dl_{-e'_2}(x - x') \bigr) \bigr].
\label{eq:pacta-quae-turpem} 
\end{align}
Exploiting 
$\del^\ka \dl_{-e'_i}(x - x') = -\del_{x'}^\ka \dl_{-e'_i}(x - x')$ as
well as skewsymmetry of $f_{dfc}$ and Maxwell's equations for
$F'^{\ka\la}_f$, we see that terms of the type
\begin{equation}
A'^{2'}_{\la d} F'^{\ka\la}_f\, \del_\ka\,\dl_{-e'_1}(x - x')
= -\del'_\ka \bigl[ 
A'^{2'}_{\la d} F'^{\ka\la}_f \,\dl_{-e'_1}(x - x') \bigr]
\end{equation}
form a divergence of an expression supported at the exceptional
set~$\bD_2$. Integrating by parts in the first line
\eqref{eq:pacta-quae-turpem}, the obstruction reads, up to a
divergence supported at~$\bD_2$,
\begin{subequations}
\label{eq:dura-lex-sed-lex} 
\begin{align}
& -i \frac{g^2}{4} f_{abc} f_{dfc} \bigl[
- \del^\la u^1_a A^2_{\nu b} A^{1'\nu}_d  A^{2'}_{\la f}
- u^1_a \del^\la A^2_{\nu b} A^{1'\nu}_d  A^{2'}_{\la f} 
+ \del^\ka u^1_a A^2_{\nu b} A^{1'}_{\ka d} A^{2'\nu}_f
\notag \\
&\hspace*{6em} 
+ u^1_a \del^\ka A^2_{\nu b} A^{1'}_{\ka d} A^{2'\nu}_f
+ u^1_a A^2_{\nu b} A^{2'}_{\la d} F^{\nu\la}_f 
- u^1_a A^2_{\nu b} A^{1'}_{\ka d} F^{\ka\nu}_f \bigr]\,\dl(x - x')
\notag \\
&= -i \frac{g^2}{4} f_{abc} f_{dfc} u^1_a \bigl[
F_b^{\mu\nu} A^{1'}_{\mu d} A^{2'}_{\nu f}
+ A^2_{\mu b} A^{2'}_{\nu d} F_f^{\mu\nu} 
+ A^2_{\mu b} A^{1'}_{\nu d} F_f^{\mu\nu} \bigr] \,\dl(x - x')
\label{eq:dura-lex} 
\\
&\quad -i \frac{g^2}{4} f_{abc} f_{dfc} d_{e_1} \bigl[
A^{1\mu}_a A^2_{\nu b} A^{1'}_{\mu d} A^{2'\nu}_f
- A^{1\mu}_a A^2_{\nu b} A^{1'\nu}_d A^{2'}_{\mu f} \bigr]
\,\dl(x - x').
\label{eq:sed-lex} 
\end{align}
\end{subequations}
We have used that $\del^\mu u^1_a = d_{e_1} A^{1\mu}_a$. On 
symmetrizing in the variables $e_2,e'_1,e'_2$, the
line~\eqref{eq:dura-lex} becomes proportional to 
\begin{equation}
- i\,g^2 u^1_a F_b^{\mu\nu} \bigl[ A^{1'}_{\mu d} A^{2'}_{\nu f}
+ A^2_{\mu d} A^{2'}_{\nu f} + A^2_{\mu d} A^{1'}_{\nu f} \bigr]
\bigl[ f_{abc} f_{cdf} + f_{afc} f_{cbd} + f_{adc} f_{cfb} \bigr]
\,\dl(x - x').
\end{equation}
It follows that the Jacobi identity:
\begin{equation}
f_{abc} f_{cdf} + f_{afc} f_{cbd} + f_{adc} f_{cfb} = 0
\end{equation}
is a \textit{necessary} condition for the obstruction to string
independence to vanish.

\subsection{The quartic term}
\label{eq:casus-foederis}

We still have to deal with the term of the type
$\del uAAA \sim d_{e_1} AAAA$ in~\eqref{eq:sed-lex}. Using the
skewsymmetry of the~$f_{abc}$, symmetrization of~\eqref{eq:sed-lex} in
the variables $e'_1$, $e'_2$ yields the identical symmetrization of
\begin{equation}
d_{e_1} \Bigl[ -i \frac{g^2}{2} f_{abc} f_{dfc} A^\mu_a(x,e_1)
A^\nu_b(x,e_2) A_{\mu d}(x,e'_1) A_{\nu f}(x,e'_2) \Bigr]
\,\dl(x - x').
\label{eq:pacta-sunt-servanda} 
\end{equation}

Thus, to keep string independence, a summand
\begin{equation}
i\,\frac{g^2}{2}\, f_{abc} f_{dfc} A^\mu_a(x,e_1) A^\nu_b(x,e_2) 
A_{\mu d}(x,e'_1) A_{\nu f}(x,e'_2) \,\dl(x - x'),
\label{eq:res-sacra} 
\end{equation}
whose string-derivative $d_{e_1}$ cancels the expression
\eqref{eq:pacta-sunt-servanda}, must be added to~$S_2$. This is the
four-gluon coupling, usually attributed to the ``covariant
derivative'' present in the ``kinetic'' QCD Lagrangian. Now, since in
the present dispensation an $AAAA$ term is also renormalizable to
begin with, we could have introduced it from the outset. Then a
discussion parallel to the above leads again to
Eq.~\eqref{eq:res-sacra} with precisely the same second-order
coefficient in the coupling constant.

\medskip

In summary: string independence of the $\bS$-matrix at second order
holds \textit{if and only if} the Jacobi identity with completely
skewsymmetric $f_{abc}$ for the cubic coupling~\eqref{eq:vetustas}
holds \textit{and} the above quartic term~\eqref{eq:res-sacra} is
present at that order in the $\bS$-matrix.

\section{Discussion}
\label{sec:nemine-contradicente}

The outcome of the previous arguments, together with the lessons on
QFD in~\cite{Rosalind}, is that Lie algebra structures of the compact
type of necessity govern interactions in QFT. Compactness is related
to complete skewsymmetry: a finite-dimensional Lie algebra structure
with generators $X_a$ defined by $[X_a, X_b] = \sum_c f_{abc} X_c$
does require $f_{abc} = -f_{bac}$ and the Jacobi identity, but not
$f_{abc} = -f_{acb}$ in general. The extra requirement imposed by
string independence leads to a negative definite Killing form and thus
semisimple compactness -- see for instance
\cite[Sect.~3.6]{DuistermaatKolk}.

The authors presently know of \textit{five} different arguments within
perturbative QFT for the reductive Lie algebra structure of the
interactions: the already classical one by Cornwall, Levin and
Tiktopoulos from unitarity bounds at high energy~\cite{CLT}; the
Aste--Scharf analysis referred to in
subsection~\ref{ssc:ad-inquirendum}; the one in this paper, and two
more found in the book~\cite{MatthewEvangile}: one in its Chapter~27,
by elaborating on the spinor-helicity formalism for gluon scattering
calculations, and another -- most charming of them all -- in its
Chapter~9, by a variant of Weinberg's ``soft limit'' reasoning, long
ago applied to link helicity~$1$ particles with charge conservation
and helicity~$2$ particles with universality of gravity. The
irrelevance of the gauge condition in the old direct construction of
scattering amplitudes by Zwanziger~\cite{Zwanziger64} also comes to
mind.

This is a good place to take stock of the lessons from
our~\cite{Rosalind}. Our argument there was also motivated, in a
somewhat contrarian way, by Marshak's thoughts on the ``neutrino
paradigm'' in his posthumous book \textit{Conceptual Foundations of
Modern Particle Physics} \cite[Chaps.~1,6]{REM} -- 
see~\cite{Belinda} as well. The sole inputs of our treatment of
flavourdynamics in it are the \textit{physical} particle types, masses
and charges: ``spontaneous symmetry breaking'', which comes in succour
of the conventional gauge models, is not required, since the SLF
models are renormalizable to begin with, irrespectively of mass.
String independence at first and second order rules the bosonic
couplings, just as in gluodynamics; the non-vanishing masses just
bring minor complications. Couplings between bosons and fermions
\textit{a~priori} are just asked to respect electric charge
conservation and Lorentz invariance. Chirality of the couplings is the
outcome of general string independence. The proof requires the
presence of the scalar particle and is done with Dirac fermions: from
the standpoint of SLF it does not make sense to say that lepton or
quark currents are ``chiral'' in QFD: their couplings~are.%
\footnote{In the current parlance, fermions in the Standard Model are
schizophrenic: non-chiral in QCD, chiral in QFD. To be sure, from the
SLF treatment one can reengineer the GWS model and its warts,
hypercharges and all, in reverse. This was done in
\cite[App.~D]{Rosalind}.}

\medskip

The cited book by Marshak is still a very good guide for the
discussion that follows -- witness perhaps to the lack of progress and
``deeply disturbing features'' \cite{TDLee} affecting the present
theoretical apparatus, contrasting with tremendous progress in the
experimental realm during recent years.

\subsection{Story of two principles}
\label{sec:hanku}

The ``gauge principle'', a \textit{top-down, classical-geometrical}
principle which has ruled particle theory for over sixty years, is
foreign to QFT. Looking back, a defect -- the unavailability in
conventional QFT of a Hilbert space framework for massless particles
-- was elevated into a doctrine. By now \textit{bottom-up, inherently
quantum} principles for the construction of interactions deserve their
place in the sun.

If only for sound epistemological reasons, one should ascertain
whither the bottom-up approach leads, in present-day particle physics.
A Lie algebra and a local Lie group amount to the same thing. Now, the
mathematical beauty and power of spin-offs like extended field
configurations in classical Yang--Mills geometry is not to be denied
-- it is enough to recall~\cite{TheDonald}. And in turn, it is
perfectly legitimate to look for \textit{global} Lie group features in
modern, bottom-up QFT. However, we contend that those need to be
constructed as \textit{quantum} field theory entities -- and, as
anyone who has had to grapple with (say) the rigorous definition of
Wilson loops for interacting fields in SLF formulations well knows,
this is more intricate than presuming classical geometry entities to
possess non-perturbative quantum counterparts.

Consider, for instance, Dirac's ad-hoc non-integrable phase and
magnetic monopole (and their progeny of `t~Hooft--Polyakov monopoles).
The basic idea was seductive, and led directly and elegantly to
electric charge discreteness. Nevertheless, humbler, essentially
perturbative arguments on cancellation of anomalies -- up to and
including the mixed gravitational-gauge anomaly \cite{LAG-W} --
recalled by Marshak -- see \cite[Chap.~7]{REM}
and~\cite{GengMarshakII}, among others~\cite{Chiron} -- are known to
provide an explanation for that discreteness
\cite[Sect.~30.4]{MatthewEvangile}. Whereas magnetic monopoles of any
kind have stubbornly refused to show~up.

This is perhaps the place to comment on the Aharonov--Bohm and
Aharonov--Casher effects being held as ``proofs'' of the
``fundamental'' character of the standard electromagnetic gauge
potential -- since the calculation via the electromagnetic field
depends on a region where the test particle is not allowed. This is
merely a misunderstanding: these effects can be computed by means of
the SL $A$-field, which contains the same information as the
$F$-field. The deep reasons lie in the mind-boggling entanglement
properties of QFT, as compared to ordinary quantum mechanics --
concretely in the failure of \textit{Haag duality} for all quantum
massless fields with helicity $r \geq 1$. This was shown over forty
years ago~\cite{LRDaniel}. Consult \cite{SchroerOnJordan} in this
respect as~well.

\subsection{Reassessing the Okubo--Marshak argument}
\label{ssc:bis-puniri}

Marshak was very open to topological and differential-geometric
constructions in QFT, and actually Chapter~10 of his book~\cite{REM},
particularly Section~10.3 and Subsection~10.3.c, is still a very good
place to learn about instantons and ``vacuum tunneling'' into
topologically inequivalent vacua, apparently leading to the degenerate
$\theta$-vacuum, the ``strong $CP$ problem'' -- since the neutron's
electric dipole moment is vanishingly small -- and, according to some,
to the ``axion''.

The steps to the claim of existence of such a problem are well known.
In the Euclidean setting first, finiteness of the classical action
functional of course requires 
$\lim_{|x|\upto\infty} F_{\mu\nu}(x) = 0$. For which it is enough to
demand that 
$\lim_{|x|\upto\infty} A_\mu(x) \to U^\7(x)\,\del_\mu U(x)$. It is 
then said that $A$ is ``pure gauge''. By using an $A_0$ gauge, a
winding number~$n$ is related directly to the Euclidean action -- in
fact, one is dealing with the homotopy group of the $3$-dimensional
sphere, which is the group of integer numbers. The next step is to
define the vacuum states in Minkowski space in such a way that
instantons become ``tunneling events'' between two different Minkowski
vacua $\ket{m}$, $\ket{m'}$ with respective winding numbers $m,m'$
satisfying $m' - m = n$. Then it is argued that the ``true vacuum'' is
of the form $\ket{\theta} = \sum_m e^{-im\theta} \ket{m}$ with
$\braket{\theta'}{\theta} = 2\pi \dl(\theta - \theta')$; and the value
of~$\theta$ is anyone's guess. This is aptly described by Marshak
in \cite[Subsect.~10.3.c]{REM}.

Nevertheless, he came to regard the axion hypothesis as a bridge too
far. In~\cite{VerdaderoKO} it was rigorously proved that the BRS
charge ``kills'' the \textit{physical} vacuum, which if cyclic must be
unique. But that charge and the antiunitary operator for CPT
invariance commute, and this obviously demands the zero (or $\pi$)
value for the $\theta$-parameter of the instanton makeup~\cite{OM92}.
One could contend that the $\theta$-vacua are non-normalizable (a sign
of trouble in itself) and that the physical vacuum is a superposition
of them. However, by means of a simple procedure, Okubo and Marshak
showed how in that case CPT invariance still guarantees the
experimentally measurable value of~$\theta$ to be zero.

The existence of the original ``visible'' axion had been already
disallowed by experiment and observation~\cite{GoodOldRaffelt} by the
time of the writing of~\cite{REM}, leading to a succession of
``invisible'' axions, that must be ``super-light'', but still face
experimental limitations. Marshak concludes: ``It does seem that the
odds of finding the `invisible' axion are rapidly diminishing and that
the incentive to carry on the ingenious searches for the `invisible'
axions is fueled more by astrophysical and cosmological interest than
by any hope of salvaging the Peccei--Quinn-type solution of the
`strong CP problem' in QCD.'' This rings even truer 29~years later.

Making a contention to the same purpose in SLF theory is
straightforward. The string-localized vector potentials live on
Hilbert space and act cyclically on the vacuum. In fact, every local
subalgebra of operators enjoys this property (this is part of the
Reeh--Schlieder property)~\cite{Witten}). Therefore $\theta$-vacua are
not allowed.

In plainer language: $F^{\mu\nu}(x) \downto 0$ at spatial infinity
implies $A^\mu(x,e) \downto 0$. Thus only the $m = 0$ vacuum (so to
speak) occurs: there are no instantons, and no $\theta$ vacua, and the
so-called strong $CP$~problem is solved. The moral of this part of the
story: quantum field theory should stand on its own feet, rather than
on classical crutches.

\subsection{Coda}
\label{sss:debitor-sui-ipsum}

The ``strong CP problem'' and the axion ideas have undergone mutations
from the time of their inception to the present day. Relations of the
present-day ``invisible'' axion with the $U_A(1)$ anomaly remain murky
-- for a recent review of the latter consult~\cite{Vicente21}. The
contemporary main selling point is still the ``axiverse''; in other
words, the search for the axion is essentially model-free nowadays.
Absence of evidence is not evidence of absence, to be sure -- and so
hunting for ALPs is bound to go on~\cite{Susana}.

A weak point of most analyses on the present subject is that
confinement is not taken into account. However, it does appear to be
inimical to violation of CP invariance~\cite{NakamuraS}.

\subsection*{Acknowledgements}

During the inception and writing of this article, CG received
financial support from the Studienstiftung des deutschen Volkes~e.V.
JM received financial support from the Brazilian research agency CNPq,
and is also grateful to CAPES and Finep. JMG-B received support from
project C1-051 of the Universidad de Costa Rica, MINECO/FEDER project
PGC2018-095328-B-I00, and the COST action CA\,18108. He still
treasures a reprint of~\cite{OM92} handed to him by Robert E. Marshak
during the latter's only visit to Spain, in the fall of 1992,
shortly before his untimely death. We are grateful to E.~Alvarez,
P.~Duch, A.~Herdegen, K.-H.~Rehren, B.~Schroer, I.~Todorov and
J.~C.~Várilly for comments, discussions and helpful suggestions.

\appendix

\section{Time ordering outside the string diagonal}
\label{app:string-chopping}

For a vertex coupling $S_1$ of the form \eqref{eq:vetustas}, we show
$S_2 := \T\bigl( S_1(x,e_1,e_2) S_1(x',e_1',e_2') \bigr)$ is fixed
outside the string diagonal~$\bD_2$ from \eqref{eq:exceptional-set} by
the causal factorization property
\begin{equation}
\T(S_1 S_1') = S_1 S_1' \word{or} S_1' S_1 \,,
\label{eq:CausalFactor} 
\end{equation}
depending on the time ordering of the respective localization regions.

In~\cite{Atropos} the corresponding statement has been shown for the
case of \textit{first degree} string-localized fields by ``chopping''
the strings, and it was indicated why the construction runs into
difficulties for higher-degree Wick polynomials. In the present case,
we have the lucky circumstance that the relevant Wick products are
just ordinary products, see~\eqref{eq:NoWick}, and we are back to the
first-degree case. We give the argument in detail.

Let $(x,e_1,e_2;x',e'_1,e'_2) \in \bD_2$ with $e_1 \neq e_2$,
$e'_1 \neq e'_2$. The localization region of $S_1(x,e_1,e_2)$ are the
two strings $x + \bR^+ e_k$, each of which we chop into a small
compact ``head'' $\sR_k$ containing the vertex~$x$ and an infinite
``tail''~$\sS_k$: 
\begin{equation}
x + \bR^+ e_k =  \sR_k \cup \sS_k,  \word{where}
\sR_k := x + [0,s_k]\,e_k, \quad \sS_k := x + [s_k,\infty)\,e_k \,.
\end{equation}
Accordingly, the string-localized potential decomposes as
$A_\mu(x,e_1) = A_\mu^\head(x,e_1) + A_\mu^\tail(x,e_1)$, where
$A_\mu^\head$ is localized in the compact ``head''~$\sR_1$ and
$A_\mu^\tail$ in the ``tail''~$\sS_1$. Same with $A_\mu(x,e_2)$. Then
the vertex coupling $S_1$ is a linear combination of four terms:
\begin{align}   
S_1 = S_1^{\head\head} + S_1^{\head\tail} 
+ S_1^{\tail\head} + S_1^{\tail\tail},  \qquad 
S_1^{\rK\rL} := \frac{g}{2} f_{abc}
\wick:A^\rK_{\mu a}(e_1) A^\rL_{\nu b}(e_2) F^{\mu\nu}_c: \,. 
\end{align}
The terms $S_1^{\tail\head}$ and $S_1^{\tail\tail}$ can be written as
\begin{equation}
S_1^{\tail\head} = \frac{g}{2} f_{abc} A^\tail_{\mu a}(e_1)\, 
\wick:A^\head_{\nu b}(e_2) F^{\mu\nu}_c:,  \qquad
S_1^{\tail\tail} = \frac{g}{2} f_{abc} A^\tail_{\mu a}(e_1)\,
A^\tail_{\nu b}(e_2) F^{\mu\nu}_c
\label{eq:NoWick} 
\end{equation}
and similarly for $S_1^{\head\tail}$. We have taken 
$A^\tail_{\mu a}(e_1)$ and $A^\tail_{\nu b}(e_2)$ out of the Wick
product; this can be done because all contractions between
$A_{\mu a}$, $A_{\nu b}$ and $F^{\mu\nu}_c$ are zero since the
indices $a,b,c$ are distinct by skewsymmetry of~$f_{abc}$. 

The localization regions of these terms are as follows. The operator
product in $S_1^{\tail\tail}$ is the ordinary product of three linear
fields: one localized in the string $\sS_1$, another in~$\sS_2$ and a
third at~$x$. $S_1^{\tail\head}$ is the product of a linear field
localized in the string~$\sS_1$ and a Wick product localized in the
compact interval~$\sR_2$ that can be made arbitrarily small
around~$x$. Similarly for~$S_1^{\head\tail}$. The term
$S_1^{\head\head}$ is a Wick product localized in $\sR_1 \cup \sR_2$
which also can be made arbitrarily small around~$x$.

We decompose $S_1(x',e_1',e_2')$ in like manner. Then the second order
$S_2$ is a sum of terms of the form
$\T\bigl[ S_1^{\rK\rL}(e_1,e_2)\, S_1^{\rK'\rL'}(e'_1,e'_2) \bigr]$,
where the operator in brackets is the product of string-localized
linear fields and (almost) point-localized Wick products. By our
hypothesis that $e_1 \neq e_2$ and $e'_1 \neq e'_2$, their
localization regions are mutually disjoint. By
\cite[Prop.~2.2]{Atropos}, such regions can be chronologically
ordered, eventually after chopping them into segments, and the
corresponding operators can be time-ordered according
to~\eqref{eq:CausalFactor}. As in the proof of in
\cite[Prop.~3.2]{Atropos}, one sees that the result is Wick's
expansion, where the time-ordering appears only within two-point
functions. (In contrast to the case considered in~\cite{Atropos}, here
there are products of time-ordered two-point functions, but these are
well-defined since they have disjoint arguments.) Again by
skewsymmetry of $f_{abc}$, only contractions between fields localized
on $e$- and $e'$-strings occur. Therefore, the restriction
$e_1 \neq e_2$ and $e'_1 \neq e'_2$ can be removed. 
The proof is complete.

\end{document}